%% file: spectralcs.tex
\documentclass{article}
\input{stdinc.tex}

\usepackage{mdwlist}
\usepackage{array}

\renewcommand{\cref}{\Cref}

\title{An explicit sparse recovery scheme in the $\ell_1$-norm}
\author{Arnab Bhattacharyya and Vineet Nair\\[1em]
{\small Indian Institute of Science}\\[-.2em]
{\small Bangalore, India}\\[-.2em]
{\small \texttt{\{arnabb, vineet.nair\}@csa.iisc.ernet.in}}}

\newcommand{\vecx}{\bm{x}}
\newcommand{\vecy}{\bm{y}}
\newcommand{\vecz}{\bm{z}}

\date{\today}

\begin{document}
\maketitle

\begin{abstract}
Consider the {\em approximate sparse recovery} problem: given $A\vecx$, where $A$ is a known $m$-by-$n$ dimensional matrix and $\vecx \in \R^n$ is an unknown (approximately) sparse vector, recover an approximation to $\vecx$. The goal is to design the matrix $A$ such that $m$ is small and recovery is efficient.  Moreover, it is often desirable for $A$ to have other nice properties, such as explicitness, sparsity, and discreteness. 

In this work, we show that we can use spectral expander graphs to explicitly design binary matrices $A$ for which the column sparsity is optimal and there is an efficient recovery algorithm ($\ell_1$-minimization). In order to recover $\vecx$ that is close to $\delta n$-sparse (where $\delta$ is a constant), we design an explicit binary matrix $A$ that has $m = O(\sqrt{\delta} \log(1/\delta)\cdot n)$ rows and has $O(\log(1/\delta))$ ones in each column. Previous such constructions were based on unbalanced bipartite graphs with high vertex expansion, for which we currently do not have explicit constructions.  In particular, ours is the first explicit non-trivial construction of a measurement matrix $A$ such that $A\vecx$ can be computed in $O(n \log(1/\delta))$ time. 
\end{abstract}
\newpage

\section{Introduction}

\subsection{Background}

High school linear algebra teaches us that given an $m$-by-$n$ matrix $A$ and a vector $\bm{y} \in \R^m$, if $m<n$, there can be infinitely many vectors $\vecx$ such that $A\vecx = \bm{y}$. So, it comes as a surprise that under a natural assumption, it is actually possible to recover $\vecx$ given $A\vecx$, and even to do so efficiently, with $m$ much smaller than $n$. The assumption is that the unknown vector $\vecx$ is {\em sparse}, meaning that most of its components are zero. This phenomenon has become an intense subject of research over the last decade. The vector $A\vecx$ is called the {\em measurement vector} or {\em sketch} of $\vecx$.

The study of sparse signal recovery has several motivations. One important application of the theory is to {\em compressive sensing}, a fast developing area in digital signal processing. Here, $\vecx$ is a sparse signal that is being sensed (such as images or audio), and $A\vecx$ is the measurement made of the signal by the sensor. The measuring apparatus directly measures $A\vecx$ instead of $\vecx$, a much smaller vector if $m \ll n$. The linearity of the measurement process is justified by the fact that one can often design hardware to implement taking the inner product of two real vectors. 
Another application of sparse approximation is to {\em data stream algorithms}. Here, the input is presented as a stream, and $\bm{x}$ denotes the frequency vector of the stream (i.e., $x_i$ is the total number of times the $i$'th item has been seen in the stream). The dimension of $\bm{x}$ is the total number of possible items, and so, for a streaming algorithm, it is preferable to maintain a sketch $A\vecx$ instead of $\vecx$ itself with $m \ll n$. The linearity of this compressing scheme makes updates to the frequency vector easy to maintain. A third application is to {\em combinatorial group testing}. Here, the vector $\vecx$ represents a universe of $n$ items, out of which some unknown $s$ are {defective} in some regard. One is allowed to conduct tests on chosen subsets of the items, where each test reveals the number\footnote{In the boolean setting, each test reveals whether there exists a defective item in the setup. The work here considers multiplication over the reals.} of defective items in the subset. The goal is to conduct a small number of tests to find out the defective set of items. If the tests are done nonadaptively (which is often a natural requirement), then the experiment consists of multiplying the unknown $\vecx$ with a binary matrix with a small number of rows and then attempting to recover $\vecx$ from the product.

We now formally define the {\sc approximate sparse recovery} problem. Let us call a vector $\vecx\in \R^n$ {\em $s$-sparse} if it has at most $s$ non-zero elements. Also, we define $\sigma_s(\vecx)_p$ to be $\min \{\|\vecx-\bm{z}\|_p: \bm{z} \text{ is }s\text{-sparse}\}$. Note that $\sigma_s(\vecx)_p$ is achieved by the $\vecz$ that consists of the $s$ largest (in magnitude) components of $\vecx$.

\begin{definition}[Approximate Sparse Recovery]\label{def:rec}
\footnote{One somewhat unusual feature in this definition is that we bound the $\ell_1$ norm of the error vector $\bm{e}$ instead of the $\ell_2$ norm. This is for for purposes of analysis. 
%\cite{BerindeGIKS08, BerindeI08} consider our version of the problem and get good results for real data using their solution. We expect the same for ours.
}
Fix reals $C_1, C_2 > 0$, a matrix $A \in \R^{m \times n}$, an integer $s\geq 1$ and reals $p,q,\eta \geq 0$. The {\em $(C_1,C_2)$-{\sc approximate sparse recovery} problem for $A$ with sparsity $s$, noise $\eta$, and $\ell_p/\ell_q$-guarantee} is the following: given a vector $\vecy \in \R^m$ where $\vecy = A\vecx + \bm{e}$ for unknown vectors $\vecx, \bm{e} \in \R^n$ with $\|\bm{e}\|_1 < \eta$, find a vector $\vecz \in \R^n$ such that:
\begin{equation}\label{eqn:obj}
\|\vecz-\vecx\|_p \leq C_1 \cdot \sigma_s(\vecx)_q + C_2 \cdot \eta
\end{equation}
If $p=q=1$, then their mention will be omitted. The matrix $A$ is called the {\em measurement matrix}.
\end{definition}

Thus, if $\vecy = A\vecx$ where $\vecx$ is $s$-sparse and $\eta = 0$, then an algorithm solving the above problem is required to output $\vecx$ itself.  In addition, our definition accomodates the cases when $\vecx$ is not exactly sparse and when $\vecy$ is not exactly equal to $A\vecx$ (due to a bounded measurement error $\eta$). Note that the recovery algorithm is required to satisfy the guarantee \cref{eqn:obj} for {\em all}\footnote{There has been other work, mostly in the data streaming community, that gives a probabilistic guarantee for each vector $\vecx$, but we will focus only on the above stronger formulation henceforth.} possible $\vecx$.

In this work, we will concentrate on the {\em $\ell_1$-minimization} or {\em basis pursuit} recovery algorithm:
\begin{align}\tag{P1}\label{eqn:l1}
\min \|\vecz\|_1 & &\text{subject to} & & \|\vecy-A\vecz\|_1 < \eta
\end{align}
The above minimization problem can be cast as a linear program and can be solved in polynomial time. The breakthrough works which sparked the modern research in sparse recovery \cite{CandesR06, CandesT05, CandesT06, Donoho06, CandesRT06a, CandesRT06b, Donoho06b} showed that if $A$ is a ``random matrix" (e.g., a matrix with i.i.d. Gaussian entries or with a random subset of the rows from the discrete Fourier transform matrix \cite{CandesRT06b}), then $\ell_1$-minimization solves the approximate sparse recovery problem with $\ell_2/\ell_1$-guarantee. Note that these matrices are dense, and there is also an issue of precision in representing real number matrix entries. 

On the other hand, we often want the measurement matrix $A$ to be binary and sparse.  In compressive sensing, the encoding time (i.e., the time needed to compute the compression $A\vecx$ when given $\vecx$) is determined by the sparsity of $A$, and discreteness of the matrix $A$ makes the implementation in hardware of the measuring device more feasible. In data stream processing, when a new item arrives (i.e., when the existing frequency vector $\vecx$ needs to be incremented by $\bm{e}_i$ for some $i \in [n]$), the time needed to update the sketch is determined by the maximum number of ones in a column of $A$. In group testing, the matrix $A$ necessarily needs to be binary, and sparsity is also helpful because the tests may not be accurate when they are conducted on a very large pool. Moreover, if the interior-point method is used to solve the program (\ref{eqn:l1}), a sparse matrix $A$ also speeds up the recovery process as the interior-point algorithm repeatedly multiplies $A$ to vectors.

An additional property we often want measurement matrices to satisfy is {\em explicitness}, i.e., constructible in time polynomial in $n$. In \cite{CandesRT06b}, the matrices $A$ are generated at random, and there is no known efficient algorithm for verifying that a matrix does indeed have the restricted isoperimetry property needed for the analysis. This is an issue for compressive sensing, since the matrix will be implemented in hardware and will sense a huge number of signals over time. Also, from a theoretical standpoint, it is an interesting question to understand whether the analytic properties required of measurement matrices can be achieved deterministically. A blog post by Tao \cite{Taoexplicitblog} highlights these issues.

\subsection{Our result}
Below, we show the current best results for binary, sparse and explicit measurement matrices recoverable using $\ell_1$-minimization, when the sparsity is $\delta n$. Here, $c > 1$ and $\eps > 0$ are some fixed constants, and we ignore multiplicative constants. \\[.5cm]

\begin{tabular}{|c|c|c|c|c|}
\hline
Paper & Sketch length & Matrix sparsity & Guarantee\\
\hline
\cite{GuruswamiLR10} $\left(\delta \leq \frac{1}{(\log n)^{c\log \log \log n}}\right)$&   $\delta^{1/\log \log n} \cdot n \log \log n$ & $n^{2-\epsilon}$ & $\ell_2/\ell_1$ \\[.2cm]
\cite{BerindeGIKS08, BerindeI08, GuruswamiUV09} $\forall \alpha > 0$ & $(\delta n)^{1+\alpha} (\log n \log \delta n)^{2+\frac2\alpha}$ & $n(\log n \log \delta n)^{1+\frac1\alpha}$ & $\ell_1/\ell_1$\\[.2cm]
\hline
\end{tabular}
~\\[.1cm]
\noindent To contrast, the best
\footnote{As far as we know, with $\ell_2/\ell_1$ guarantee, the only construction of a non-trivially sparse binary measurement matrix for which basis pursuit recovery works is the explicit construction of \cite{GuruswamiLR10}. Also, \cite{CandesRT06b} showed that $\ell_2/\ell_1$ guarantee holds for a random Bernoulli matrix, where each entry is $\pm 1$ with equal probability, but such a matrix is dense.} 
non-explicit construction of a sparse, binary measurement matrix recoverable using $\ell_1$-minimization is:\\[.2cm]
\begin{tabular}{|>{\centering\arraybackslash}p{4.9cm}|>{\centering\arraybackslash}m{3.8cm}|>{\centering\arraybackslash}m{2.8cm}|>{\centering\arraybackslash}m{1.5cm}|}
\hline
\cite{BerindeGIKS08, BerindeI08} & $\delta \log(1/\delta) \cdot n$ & $\log(1/\delta)\cdot n$ & $\ell_1/\ell_1$\\[.2cm]
\hline
\end{tabular}
~\\[.2cm]
Our main result is:
\begin{theorem}\label{thm:main}
Fix a constant $\delta \in (0,1)$. Then, for all large enough $n$, there exists a binary matrix $A \in \{0,1\}^{m \times n}$ that can be explicitly constructed and  has the following properties:
\begin{itemize}
\item
The sketch length $m = c\sqrt{\delta} \log(1/\delta) \cdot n$.
\item
$A$ has $c' \log(1/\delta)\cdot n$ nonzero entries.
\item
There exist $C_1, C_2 > 0$ such that for all $\eta \geq 0$, the $(C_1, C_2)$-{\sc approximate sparse recovery} problem for the matrix $A$  with sparsity $\delta n$, noise $\eta$ and $\ell_1/\ell_1$ guarantee can be solved using the $\ell_1$-minimization algorithm.
\end{itemize}
Here, $c$ and $c'$ are independent of $\delta$ and $n$, while $C_1$ and $C_2$ are independent of $n$.
\end{theorem}

To construct $A$, we use the same basic idea as \cite{GuruswamiLR10}, as we explain in the next section. Somewhat surprisingly, we show that if we care about an $\ell_1/\ell_1$-guarantee, instead of\footnote{In \cite{CohenDD09}, it's shown that ``$\ell_2 \leq \frac{C}{\sqrt{s}}\ell_1$'' is a stronger guarantee than ``$\ell_1 \leq (1+O(C)) \ell_1$''. } an $\ell_2/\ell_1$-guarantee as Guruswami et al. \cite{GuruswamiLR10} do, the sparsity required for the matrix drops to what is required for the best non-explicit construction \cite{BerindeI08}.  A limitation of our result is that it only applies to when $\delta$ is a constant, whereas the other results cited above work for subconstant $\delta$ (and in fact, \cite{GuruswamiLR10} does not apply for linear sparsity). However, the case of constant $\delta$ is of interest as well. 

\ignore{Our proof of \cref{thm:main} also applies to some extensions of \cref{def:rec}. Suppose that for the error vector $\bm{e}$ in \cref{def:rec}, we have $\|\bm{e}\|_s < \eta$ for some arbitrary $s \geq 0$. Naturally, we also modify the $\ell_1$-minimization program (\ref{eqn:l1}) to require that $\|\vecy-A\vecz\|_s < \eta$. Then, \cref{thm:main} still holds true. Of course, the $\ell_1$-minimization problem may not be solvable in polynomial time for arbitrary norms, but the setting where $s=1$ is commonly considered and here, the modification of (\ref{eqn:l1}) reduces to solving a linear program. Indeed, Berinde et al. \cite{BerindeGIKS08, BerindeI08} restrict themselves to the case $\|e\|_1 < \eta$, whereas our analysis is more flexible towards the norm used for bounding the error vector. }

\subsection{Relation to Coding Theory and Expander Graphs}

Our construction of the measurement matrix is inspired by a natural analogy between sparse recovery and linear error-correcting codes over a finite alphabet. One can view sparse recovery as error correction over the reals in the following sense. Just as a linear error-correcting code $\mathcal{C}$ over $\{0,1\}^n$ of distance $d$ has the property that all nonzero codewords of $\mathcal{C}$ have weight at least $d$, sparse recovery is possible only if the kernel of the measurement matrix $A$ does not contain any nonzero vector with $\leq 2s$ nonzero entries. For otherwise, we would have two $s$-sparse vectors $\vecx_1$ and $\vecx_2$ with $A\vecx_1 = A\vecx_2$ and so, recovery of $\vecx_1$ given $A\vecx_1$ would be impossible. The connection between sparse recovery and error-correcting codes is of course quite well-known; see Cheraghchi's \cite{Cheraghchi11} for a more detailed study of the interplay between the two subjects. 

In fact, the works \cite{BerindeGIKS08, BerindeI08, GuruswamiLR10} cited above achieving the current best results for explicit, sparse measurement matrices also use this connection. To implement this connection, these works (and the present one) devise analogs of {\em expander codes} for sparse recovery. \ignore{An {\em expander graph}, in our context, is a bipartite graph with $n$ vertices on the left, $m$ vertices on the right, and left-degree $d$, having the property that every small set of vertices on the left expands to a proportionately ``large'' set of vertices on the right. }In the by-now classic work in coding theory, Sipser and Spielman \cite{SipserS96} showed two ways in which expander graphs can yield error-correcting codes.   

The first approach is to use a good {\em vertex expander}. For the purposes of this discussion, let this be a graph in which every subset $S$ of vertices on the left of size at most $\delta n$ has more than $\frac{3}{4}d|S|$ neighbors on the right. Then, Sipser and Spielman proved that a code whose parity check matrix is the adjacency matrix of such a vertex expander has distance more than $\delta n$ and that the decoding can be carried out in linear time. The rate of the code corresponds to $1-m/n$, and hence it is important to ask how small $m$ can be. A straightforward probabilistic argument shows that a random bipartite graph with $d = O(\log(1/\delta))$ and $m=O(\delta \log(1/\delta)\cdot n)$ has the desired vertex expansion. However, it was not known at the time of \cite{SipserS96} whether such high-quality vertex expanders could be built explicitly. Later, a sequence of works \cite{CapalboRVW02, TashmaU06, TashmaUZ07,  GuruswamiUV09} investigated this question; the current record is held by Guruswami et al. \cite{GuruswamiUV09} who showed that the desired vertex expander can be built deterministically in polynomial time with $d = (\log n \log \delta n)^{1+\frac1\alpha}$ and $m = O((\delta n)^{1+\alpha} (\log n \log \delta n)^{2+\frac2\alpha})$ for all $\alpha > 0$. It is still not known how to get explicit vertex expanders with the optimal parameters.

The second approach is to use the edge-vertex incidence graph of a {\em spectral expander}. For a $d$-regular (non-bipartite) graph $G$, we will say it is spectrally expanding if the second largest (in magnitude) eigenvalue of the adjacency matrix of $G$ is $< d^{0.9}$. Also, let $\cC_0 \subset \{0,1\}^d$ be a finite linear code, i.e., a linear subspace. Now, in a manner to be made precise in \cref{sec:prelim}, one can construct the {\em Tanner code} $X(G,\cC_0)$ with message length $n=Nd/2$ where $N = |V(G)|$.  Sipser and Spielman showed that if $G$ is a spectral expander and $\cC_0$ is a code with relative distance $\delta_0$, then $X(G,\cC_0)$ is a code of length $\bits^n$ with relative distance $\approx \delta_0^2$ and is decodable in linear time. Moreover, spectral expanders with optimal parameters {\em can} be constructed explicitly unlike vertex expanders. In fact, Margulis \cite{Mar73} and Lubotzky, Phillips and Sarnak \cite{LubotzkyPS88} showed explicit constructions of {\em Ramanujan graphs} with second eigenvalue equal to $2\sqrt{d-1}$ (which is the best possible). Using this, we get explicit constructions of linear-time decodable codes with rate approximately $1-2h(\delta_0)$, where $h$ is the binary entropy function.

Both constructions of expander codes have found parallels in the analysis of sparse recovery using basis pursuit. Berinde et al. \cite{BerindeGIKS08, BerindeI08} showed that if $A$ is the adjacency matrix of a vertex expander, then sparse recovery is possible (with $\ell_1/\ell_1$-guarantee). This, together with the aforementioned explicit construction of \cite{GuruswamiUV09}, implies the results  attributed to \cite{BerindeGIKS08, BerindeI08} in the previous section. Guruswami, Lee and Razborov \cite{GuruswamiLR10} showed that if $A$ is constructed using the Tanner code $X(G,\cC_0)$ where $G$ is a spectral expander and $\cC_0$ is a finite size measurement matrix with optimal properties, then sparse recovery can be performed for $A$ with $\ell_2/\ell_1$ guarantee\footnote{In fact, their work is primarily motivated by the geometric question of finding Euclidean sections in $\ell_1^n$, for which we do not make any progress here.}. Subsequent work by Guruswami, Lee and Wigderson \cite{GuruswamiLW08} furthered this investigation and provided better parameters but at the expense of some randomness.

Our work continues in the vein of \cite{GuruswamiLR10, GuruswamiLW08} by studying continuous versions of Tanner codes. Just as in \cite{GuruswamiLW08}, our measurement matrix is simply the parity check matrix of the corresponding Tanner code. The arguments to establish the claims of the last section are clean and simple, paralleling the proof devised by Z\'emor \cite{Zemor01} to improve upon Sipser and Spielman's analysis of decoding Tanner codes.

\subsection{On faster recovery}

Thus far we have only discussed the $\ell_1$-minimization algorithm for sparse recovery. There exist other ``combinatorial" algorithms, which typically proceed by iteratively identify and removing large components of the unknown sparse vector $\vecx$. There are two types of such results. The first type applies when $\vecx$ is known to be exactly $s$-sparse, i.e., $\sigma_s(\vecx) = 0$. Here, Indyk \cite{Indyk08}  and Xu and Hassibi \cite{XuH07} showed an explicit construction of a binary matrix $A \in \{0,1\}^{m \times n}$ such that $m = s \cdot 2^{\text{poly}(\log \log n)}$ and such that $\vecx$ can be recovered from $A\vecx$ in time $O(m \text{ polylog}(n))$. 

The second type of result works for any $\vecx$, not necessarily sparse, and solves the approximate sparse recovery problem. The best result in this direction is by Indyk and Ru\v{z}i\'c \cite{IndykR08} who showed that for any $\delta > 0$, when the measurement matrix $A$ is the adjacency matrix of a vertex expander of the appropriate size, for any $\epsilon, \eta> 0$, there is an algorithm solving the $(1+\epsilon,6)$-{\sc approximate sparse recovery} problem for $A$ with sparsity $\delta n$, noise $\eta$ and $\ell_1/\ell_1$ guarantee. The parameters of $A$ match those of the construction from \cite{BerindeGIKS08, BerindeI08}. The recovery algorithm runs in time $O(n \log n)$ in contrast to the time required for solving the $\ell_1$-minimization problem, which heuristically requires about $\tilde{O}(n^{1.5})$ time. A subsequent work by Berinde, Indyk and Ru\v{z}i\'c \cite{BerindeIR08} simplified the algorithm and made it efficient in practice, at the cost of an extra logarithmic factor in the theoretical bound for the running time.

For our construction, we note firstly that when the unknown $\vecx$ is exactly $\delta n$ sparse, an $O(n \log n)$ time iterative algorithm for recovery (with the same parameters for the measurement matrix as in \cref{thm:main}) directly follows from an observation by Guruswami et al.  in section 3.2 of \cite{GuruswamiLW08}). It is an interesting open question to have a similar algorithm for the approximate sparse recovery problem.

\section{Preliminaries}\label{sec:prelim}

\subsection{Basic notions}
We use $[n]$ to denote the set $\{1, 2, \dots, n\}$, and for a subset $S \subseteq [n]$, $\bar{S}$ denotes $[n]\setminus S$. For $\vecx \in \R^n$ and $I \subseteq [n]$, we let $\vecx_I$ denote the restriction of $\vecx$ to the coordinates in $I$.

We next define the graphs upon which our main construction is based.
\begin{definition}[Spectral expanders]
A simple, undirected graph $G$ is an {\em $(N,d,\lambda)$-expander} if $G$ has $N$ vertices, is $d$-regular, and the second largest eigenvalue of the adjacency matrix of $G$ in absolute value is at most $\lambda$.
\end{definition}

The combinatorial implication of spectral expansion is the following useful fact (see e.g. Lemma 2.5 of \cite{HooryLW06}):
\begin{lemma}[Expander Mixing Lemma]\label{expmix}
Let $G=(V,E)$ be a $(N,d,\lambda)$-expander. Then, for all $S, T \subseteq V$:
$$\left| |E(S,T)| - \frac{d|S||T|}{N} \right| \leq \lambda \sqrt{|S||T|}$$
where $E(S, T)$ is the set of edges between sets $S$ and $T$ with the edges in $(S \cap T) \times (S \cap T)$ counted twice.
\end{lemma}

Because of technical reasons, we will work with a bipartite version of a given spectral expander. More precisely:
\begin{definition}[Double Cover]
The {\em double cover} of a graph $G = (V,E)$ is the bipartite graph $H = (L\cup R, E_H)$ where $L = R = V$ and for all $(u,v) \in E$, both $(u,v)$ and $(v,u)$ are in $E_H$. 
\end{definition}
The following is immediate:
\begin{corollary}[Bipartite Expander Mixing Lemma]\label{cor:mix}
Let $H=(L \cup R,E)$ be the double cover of an $(N,d,\lambda)$-expander. Then, for all $S\subseteq L$ and $T \subseteq R$:
$$\left| |E(S,T)| - \frac{d|S||T|}{N} \right| \leq \lambda \sqrt{|S||T|}$$
where $E(S, T)$ is the set of edges between sets $S$ and $T$.
\end{corollary}

\subsection{Tanner Measurements}
Suppose $G$ is an $(N,d,\lambda)$-expander graph, and let $H = (L \cup R, E)$ be its double cover. The number of edges in $H$ is $dN$. Let us also label the edges in some arbitrary way from $1$ to $Nd$. Fix $\cC_0 \in \R^{k \times d}$ to be a $k$-by-$d$ matrix for some $k \leq d$.  Now, as in \cite{GuruswamiLR10, GuruswamiLW08}, we define a new measurement matrix as follows:
\begin{definition}
The {\em Tanner measurement matrix of $G$ and $\cC_0$} is the $2kN$-by-$dN$ matrix $A$ defined by:
$$A\vecx = (\cC_0 \vecx_{\Gamma(v)})_{v \in L \cup R}~~~\forall \vecx \in \R^{dN}$$
where $\Gamma(v)$ is the tuple of edges incident to $v$ ordered in increasing order. That is, $A\vecx$ is simply the concatenation of $\cC_0\vecx_{\Gamma(v)}$ for all vertices $v$ in the double cover graph $H$. 
\end{definition}

Note that in the above definition, $\vecx$ corresponds to a real number assignment to the edges of $H$. 

\subsection{Robust Null Space Property}
Is there an analytic condition on $A$ which can ensure that $\ell_1$-minimization succeeds in solving the sparse recovery proplem for $A$? There have been multiple attempts in formulating such a condition. Cand\'es et al. in \cite{CandesRT06b} derived the {\em restricted isoperimetry property} which sufficed to show that $\ell_1$-minimization succeeds for random Gaussians and Fourier matrices. However, Chandar \cite{Chandar08} showed that any {\em binary} matrix with $o(s^2)$ rows would not satisfy the RIP property. Berinde et al. \cite{BerindeGIKS08} evaded this negative result by formulating a variant of the RIP condition which can be satisfied by binary matrices with $O(s \log n/s)$ rows and is also a sufficient condition for the success of $\ell_1$-minimization. We use yet another condition which is especially convenient for our analysis; this formulation is from the wonderful recent book by Foucart and Rauhut \cite{FoucartR13}.

\begin{definition}
A matrix $A \in \R^{m \times n}$ satisfies the {\em robust null space property of order $s$} with constants $0<\rho<1$ and $\tau > 0$ if for all $S \subseteq [n]$ of size at most $s$ and for all $\vecx \in \R^n$:
$$\|\vecx_S\|_1 \leq \rho \|\vecx_{\bar{S}}\|_1 + \tau \|A\vecx\|_1$$
\end{definition}

 The robust null space condition ensures the following:
\begin{theorem}[Theorem 4.19 of \cite{FoucartR13}]
Suppose that $A \in \R^{m \times n}$ satisfies the robust null space property of order $s$ with constants $0<\rho<1$ and $\tau > 0$. Then, for all $\eta \geq 0$, $\ell_1$-minimization (\ref{eqn:l1}) solves the $\left(\frac{2(1+\rho)}{1-\rho}, \frac{4\tau}{1-\rho}\right)$-{\sc approximate sparse recovery problem} for $A$ with sparsity $s$, noise $\eta$ and $\ell_1/\ell_1$-guarantee.
\end{theorem}

Our main goal from now on will be to show that the robust null space property holds for the Tanner measurement matrix of $G$ and $\cC_0$ where $G$ is an $(n,d,\lambda)$-expander (with $\lambda \ll d$) and $\cC_0: \R^{k \times d}$ itself admits a satisfactorily strong null space property.

\section{Proof of Main Result}

Suppose $\cC_0 \in \R^{k \times d}$ satisfies the robust null space property of order $\delta_0 d$ with constants $\rho_0 < 1/3$ and $\tau_0 > 0$. We want that $k \leq 100\delta_0 \log(1/\delta_0) d$ and each column of $\cC_0$ contain at most $100 \log 1/\delta_0$ ones. These parameters are achievable through random constructions for sufficiently large $d$, and so, $\cC_0$ can be found by brute force if $d$ is bounded. 

For a $(n,d,\lambda)$-expander graph $G$ and $\cC_0$ as above, let $A$ be the Tanner measurement matrix of $G$ and $\cC_0$. \cref{thm:main} is a consequence of the following:
\begin{theorem}\label{thm:sec}
Given $\delta > 0$, let $d > 16/\delta$ and $\delta_0 = 2\sqrt{\delta}$.  Let $\cC_0 \in \R^{k \times d}$ be the matrix defined above with parameters $\rho_0 < 1/3$ and $\tau_0 > 0$. Let $G$ be an $(N, d, \lambda)$-expander graph, with $\lambda < 3\sqrt{d}$. Then, the matrix $A \in \R^{2kN \times dN}$, the Tanner measurement matrix of $G$ and $\cC_0$, satisfies the robust null space property of order $\delta Nd$ with constants $\rho=\frac{2\rho_0}{1-\rho_0}$ and $\tau=\frac{\tau_0}{1-\rho_0}$. 
\end{theorem}
\begin{proof}
Fix a set $S$ of size at most $\delta Nd$. We would like to show that for any $\vecx \in \R^{Nd}$:
$$\|\vecx_S\|_1 \leq \rho \|\vecx_{\bar S}\|_1 + \tau \|A\vecx\|_1$$

Let $H = (L \cup R, E_H)$ be the double cover of $G$. Recall that the edges of $H$ are distinctly labeled from $1$ to $Nd$. For a set $T \subseteq [Nd]$ and a vertex $v \in L \cup R$, let $\Gamma(v,T) = \{i : i \in T, \text{edge labeled }i\text{ is incident to }v\}$, and let $\deg(v,T) = |\Gamma(v,T)|$.

Let us first consider the case when for all $v \in L$, $\deg(v,S) \leq \delta_0 d$. We can easily establish the robust null space property in this case. 
\begin{lemma}\label{basic}
If for all $v \in L$, $\deg(v,S) \leq \delta_0 d$ or for all $v \in R$, $\deg(v,S) \leq \delta_0 d$ then $$\|\vecx_S\|_1 \leq \frac{\rho_0}{1+\rho_0} \|\vecx\|_1 + \frac{\tau_0}{1+\rho_0} \sum_{v \in L} \|\cC_0 \vecx_{\Gamma(v)}\|_1  \leq \frac{\rho_0}{1+\rho_0} \|\vecx\|_1 + \frac{\tau_0}{1+\rho_0} \|A \vecx\|_1$$
\end{lemma}
\begin{proof}
Let us assume for all $v \in L$, $\deg(v,S) \leq \delta_0 d$. In this case, for every $v \in L$, we can apply the robust null space property for $\cC_0$:
$$\|\vecx_{\Gamma(v,S)}\|_1 \leq \rho_0 \|\vecx_{\Gamma(v)\setminus \Gamma(v,S)}\|_1 + \tau_0 \|\cC_0x_{\Gamma(v)}\|_1 $$ 
Adding up these inequalities for all $v \in L$, we obtain:
$$\|\vecx_S\|_1 \leq \rho_0 \|\vecx_{\bar S}\|_1 + \tau_0 \sum_{v \in L} \|\cC_0 \vecx_{\Gamma(v)}\|_1 $$
Adding $\rho_0 \|\vecx_S\|_1$ to both sides of the above inequality we get:  
$$\|\vecx_S\|_1 \leq \frac{\rho_0}{1+\rho_0} \|\vecx\|_1 + \frac{\tau_0}{1+\rho_0} \sum_{v \in L} \|\cC_0 \vecx_{\Gamma(v)}\|_1  \leq \frac{\rho_0}{1+\rho_0} \|\vecx\|_1 + \frac{\tau_0}{1+\rho_0} \|A \vecx\|_1$$
where the last line uses the fact that by definition of $A$, we have $$\|A \vecx\|_1 = \sum_{v \in L} \|\cC_0 \vecx_{\Gamma(v)}\|_1 + \sum_{v \in R} \|\cC_0 \vecx_{\Gamma(v)}\|_1$$
The case when for all $v \in R$, $\deg(v,S) \leq \delta_0 d$ can be handled similarly 
\end{proof}

The above lemma shows that if for all $v\in L$ (or for all $v\in R$), $\deg(v,S)\leq \delta_0 d$ then we can prove the robust null space property for $A$. But the issue is this may not be true for an arbitrary choice of $S$. The idea in this case would be to create a finite sequence of sets $T_1,T_2,\dots,T_\ell$ such that $T_1 \cup T_2 \cup \cdots \cup T_\ell=S$, and for all $i,j\in[m]$, $T_i \cap T_j=\emptyset$. Moreover for all $i\in [\ell]$, either for all $v\in L$, $\deg(v,T_i)\leq \delta_0 d$ or for all $v\in R$, $\deg(v,T_i)\leq \delta_0 d$. If we are able to create such a sequence of sets, then we can use the above argument for each of the $\ell$ sets and add them to get the robust null space property.  

Let $S_0=S$. For odd $i\in[\ell]$, let $$S_i = \{e = (u,v) \mid e \in S_{i-1}, \deg(u,S_{i-1}) > \delta_0 d\}$$ and for even $i$, let $$S_i = \{ e = (u,v) \mid e \in S_{i-1}, \deg(v,S_{i-1}) > \delta_0 d\}.$$ Also, for odd $i$, define $$T_i =\{e = (u,v) \mid e \in S_{i-1}, \deg(u,S_{i-1}) \leq \delta_0 d\}$$ and for even $i$, let $$T_i = \{e = (u,v): e \in S_{i-1}, \deg(v,S_{i-1}) \leq \delta_0d\}$$  Observe that for all $i\in[m]$ $T_i=S_{i-1}\setminus S_i$ and so, the $T_i$'s are all disjoint. For odd $i$ we define the the set of vertices incident to the edges in the set $S_i$ as follows  $$V_i =\{v\in L \mid v\textnormal{ is incident to an edge in } S_i\}$$ and for even $i$, let $$V_i =\{v\in R \mid v\textnormal{ is incident to an edge in } S_i\}$$
The following lemma proves the finiteness of the sequence.
\begin{lemma}\label{zemor}
If $S_i$ is nonempty, then $|V_i|>|V_{i+1}|$.  
\end{lemma}
\begin{proof}
The proof is identical to Z\'emor's \cite{Zemor01} analysis of his decoder for Tanner codes! For completeness, we give the analysis.

First, note that $S_i \subseteq S$, and so, $|S_i| \leq \delta Nd$. On the other hand, $|S_i| > \delta_0 d \cdot |V_i|$. Hence, $$|V_i| < \frac{\delta N}{\delta_0}$$

Next, we examine $|E(V_i, V_{i+1})|$. On one hand, $E(V_i,V_{i+1}) \supseteq S_{i+1}$, and so, $|E(V_i, V_{i+1})| > \delta_0 d |V_{i+1}|$. On the other hand, by the Expander Mixing Lemma (\cref{expmix}), 
$$|E(V_i,V_{i+1}| \leq \frac{d|V_i||V_{i+1}|}{N} + \lambda\sqrt{|V_i| |V_{i+1}|} \leq \frac{d|V_i||V_{i+1}|}{N} + \lambda\frac{|V_i|+ |V_{i+1}|}{2} \leq \frac{\delta d}{\delta_0}|V_{i+1}| + \lambda\frac{|V_i|+ |V_{i+1}|}{2}
$$
where the second inequality is AM-GM and the third inequality uses the bound on $|V_i|$ from above. Combining the upper and lower bounds on $|E(V_i,V_{i+1})|$, we obtain:
$$
|V_{i+1}| < \frac{\nfrac{\lambda}{2}}{\delta_0 d-\nfrac{\delta d}{\delta_0} -\nfrac{\lambda}{2}} |V_i|
$$
Using $\delta = \delta_0^2/4$, $d > 16/\delta$, and $\lambda < 3\sqrt{d}$, we get that $|V_{i+1}| < |V_i|/3$.
\ignore{
Without loss of generality, assume $i$ is an odd number. Let $|V_{i}|=q$ and $|V_{i+1}|=t$. To prove the lemma it is sufficient to show that $\frac{t}{q}$ $<$ 1. Let $|E(V_i,V_{i+1})|$ denote the number of edges between $V_i$ and $V_{i+1}$.  The proof comes from the observation that on one hand, by definition of the sets, $|E(V_i,V_{i+1})| > \delta_0dt$, and on the other hand, by \cref{cor:mix}, $|E(V_i,V_{i+1})| \leq \frac{dqt}{N} + \lambda\sqrt{qt}$. Using the AM-GM inequality yields:
$$\frac{t}{q}(\delta_0d - \frac{\lambda}{2}) < \frac{dt}{N} + \frac{\lambda}{2}$$ 
Now recall that: $|S| \geq \delta_0 d t$ but on the other hand, $|S|\leq \delta d N\leq \delta_0 (\delta_0-\frac{2\lambda}{d})dN$. So:
$$t \leq \frac{|S|}{\delta_0 d} \leq \frac{\delta dN}{\delta_0 d} = (\delta_0 - \frac{2\lambda}{d})N$$
Putting it all together, we get:
\begin{align*} 
&\frac{t}{q} < \frac{\delta_{0}d - \frac{3\lambda}{2} }{\delta_0d - \frac{\lambda}{2}} < 1 
\end{align*}
Since $\frac{t}{q}$ $<$ 1 we have $|V_i|>|V_{i+1}|$. }
\end{proof}
Now, we can finish the proof of \cref{thm:sec}. For each $T_i$, we can separately apply \cref{basic}. For all $i$ (whether odd or even): 
$$\|\vecx_{T_i}\|_1 \leq \frac{\rho_0}{1+\rho_0}\|\vecx_{\Gamma(V_i)}\|_1 + \frac{\tau_0}{1+\rho_0}\sum_{v\in V_i}\|\cC_0 \vecx_{\Gamma(v)}\|$$
Summing over all $i$ (odd and even), we obtain:
$$\|\vecx_{S}\|_1 \leq  \frac{2\rho_0}{1+\rho_0}\|\vecx\|_1 + \frac{\tau_0}{1+\rho_0}\|A\vecx\|_1 $$
The factor $2$ appears because each edge is incident on two vertices. Rearranging, we get:
$$\|\vecx_{S}\|_1 \leq  \frac{2\rho_0}{1-\rho_0}\|\vecx_{\overline{S}}\|_1 + \frac{\tau_0}{1-\rho_0} \|A\vecx\|_1$$
\end{proof}

To complete the proof of \cref{thm:main}, we are just left to argue that the matrix $A$ has $O(\log\frac{1}{\delta})$ ones per column and $O(\sqrt{\delta}\log\frac{1}{\delta}Nd)$ many rows. We know that the number of 1's per column in $\cC_0$ is $\leq 100\log\frac{1}{\delta_0}$. An edge $e$ in the graph $H$ can be incident to at most 2 vertices. Hence, the number of 1's in any column indexed by the edges in the matrix $A$ is at most $200\log\frac{1}{\delta_0}$. Substituting $\delta=\delta_0(\delta_0-\frac{2\lambda}{d})\approx \delta_o^2$ we get the the number of ones in any column indexed by the edges in the matrix $A$ is at most $200\log\frac{1}{\delta}$. 
The total number of rows in the matrix A is $2Nk$ and we know $k$ is  $\leq 100 \delta_0 d\log\frac{1}{\delta_0}$. This implies the number of rows in the matrix A is at most $200N\delta_0d\log\frac{1}{\delta_0}$, which is approximately equal to $100\sqrt{\delta}\log\frac{1}{\delta}\cdot Nd$. 

\bibliographystyle{alpha}
\bibliography{sparse}

\end{document}

%% file: stdinc.tex
\usepackage{verbatim}
\usepackage{color}
\usepackage{graphicx}
\usepackage{tabularx}

\usepackage{amsmath}
\usepackage[showonlyrefs]{mathtools}
\usepackage{mathstyle}
\usepackage{empheq}
\usepackage{amstext,amssymb,amsfonts}
\usepackage{fullpage}
\usepackage{nicefrac}
\usepackage{bm}

\usepackage{todonotes}

\newcounter{mynotes}
\usepackage{textcomp,setspace}
\usepackage{nameref}
\usepackage[pagebackref,colorlinks,linkcolor=blue,filecolor = blue,citecolor = blue, urlcolor = blue, hyperfootnotes=false]{hyperref}
\usepackage{color}
 \usepackage[capitalize]{cleveref}
\usepackage{amsthm}
\usepackage{thmtools}

\declaretheorem[within=section]{theorem}
\declaretheorem[sibling=theorem]{corollary}

\declaretheorem[sibling=theorem]{lemma}

\declaretheorem[sibling=theorem]{definition}

\newcounter{termcounter}
\renewcommand{\thetermcounter}{\Alph{termcounter}}
\crefname{term}{term}{terms}
\creflabelformat{term}{\textup{(#2#1#3)}}

\makeatletter
\def\term{\@ifnextchar[\term@optarg\term@noarg}%]
\def\term@optarg[#1]#2{%
  \textup{(#1)}%
  \def\@currentlabel{#1}%
  \def\cref@currentlabel{[][2147483647][]#1}%
  \cref@label[term]{#2}}
\def\term@noarg#1{%
  \refstepcounter{termcounter}%
  \textup{(\thetermcounter)}%
  \cref@label[term]{#1}}
\makeatother
\newcommand{\nfrac}{\nicefrac}

\newcommand{\ignore}[1]{}

\newcommand{\bits}{\{0,1\}}

\definecolor{DSred}{rgb}{1,0,0}

\renewcommand{\leq}{\leqslant}
\renewcommand{\geq}{\geqslant}

\renewcommand{\epsilon}{\varepsilon}
\newcommand{\eps}{\epsilon}

\newcommand{\R}{\mathbb{R}}

\newcommand{\cC}{\mathcal C}